\documentclass[letterpaper,10 pt,conference]{ieeeconf}  
\IEEEoverridecommandlockouts                              

\usepackage{comment}

\usepackage{subfigure}
\usepackage{xcolor}
\usepackage{xparse}
\usepackage{mathrsfs}
\usepackage{graphics} 
\usepackage{amsmath} 
\usepackage{amssymb}  
\usepackage{comment}

\usepackage{hyperref}%
\hypersetup{colorlinks=true,  linkcolor=blue,  breaklinks=true,  urlcolor=blue,  citecolor=blue}

\usepackage{lipsum}
\usepackage{mwe}
\usepackage{soul}
\usepackage{url}

\usepackage[prependcaption, colorinlistoftodos]{todonotes}

\newcommand{\map}[3]{#1: #2 \rightarrow #3}

\newcommand{\real}{\ensuremath{\mathbb{R}}}

\newcommand{\realnonnegative}{\ensuremath{\mathbb{R}}_{\ge 0}}

\newcommand\oprocendsymbol{\hbox{$\triangle$}}
\newcommand\oprocend{\relax\ifmmode\else\unskip\hfill\fi\oprocendsymbol}

\renewcommand{\theenumi}{(\roman{enumi})}

\DeclareSymbolFont{bbold}{U}{bbold}{m}{n}
\DeclareSymbolFontAlphabet{\mathbbold}{bbold}
\newcommand{\vect}[1]{\mathbbold{#1}}

\NewDocumentCommand\norm{mg}{\lVert #1 \rVert \IfNoValueF{#2}{_{#2}}}
\NewDocumentCommand\bignorm{mg}{\big\lVert #1 \big\rVert \IfNoValueF{#2}{_{#2}}}
\NewDocumentCommand\Bignorm{mg}{\Big\lVert #1 \Big\rVert \IfNoValueF{#2}{_{#2}}}
\NewDocumentCommand\seminorm{mg}{\lvert\!\lvert\!\lvert #1 \rvert\!\rvert\!\rvert \IfNoValueF{#2}{_{#2}}}

\newtheorem{theorem}{Theorem}

\newtheorem{remark}[theorem]{Remark}

\newcommand{\HKn}{C_{N}}

\title{\LARGE\bf Multi-dimensional extensions of the Hegselmann-Krause  model}

 \author{Giulia De Pasquale and Maria Elena Valcher
 \thanks{
This work was supported in part from Fondazione Ing. Aldo Gini.
This paper has been submitted to the 61st Conference on Decision and Control (CDC 2022), Cancun, Mexico.
G. De Pasquale and M.E. Valcher are with
 the Dipartimento di Ingegneria dell'Informazione
 Universit\`a di Padova, 
    via Gradenigo 6B, 35131 Padova, Italy, e-mail:  \texttt{giulia.depasquale@phd.unipd.it, meme@dei.unipd.it}.
   } 
  }

\def\qed {{
		\parfillskip=0pt 
		\widowpenalty=10000 
		\displaywidowpenalty=10000 
		\finalhyphendemerits=0 
		\leavevmode 
		\unskip 
		\nobreak 
		\hfil 
		\penalty50 
		\hskip.2em 
		\null 
		\hfill 
		$\blacksquare$
		\par}}

\newtheorem{mythe}{Theorem}
\newtheorem{myrem}[mythe]{Remark}
\newtheorem{myexp}[mythe]{Example}
\newtheorem{mylem}[mythe]{Lemma}
\newtheorem{mydef}[mythe]{Definition}
\newtheorem{myprop}[mythe]{Proposition}

\begin{document}

\maketitle
\thispagestyle{empty}
\pagestyle{empty}


\begin{abstract}
In this paper we consider two multi-dimensional Hagselmann-Krause (HK) models for opinion dynamics.  The two   models describe how individuals  adjust  their opinions on multiple topics, based on the influence of their peers. The models differ in the criterion according to which individuals decide whom they want to be influenced from. 
 In the average-based model individuals compare their average opinions on the various topics with those of the other individuals, and interact only with those individuals whose average opinions lie within a confidence interval. For this model we  provide an alternative proof for the contractivity of the range of opinions, and show that the agents'  opinions reach consensus/clustering if and only  if their average opinions  do so. 
  In the uniform affinity model agents compare their opinions on each single topic and influence each other only if, topic-wise,   such opinions do not differ more than a given tolerance.  We    identify conditions under which the uniform affinity model enjoys the order-preservation property topic-wise and we prove that  the global range of opinions (and hence the range of opinions on each single topic) are non-increasing. 
\end{abstract}

\section{Introduction}

\emph{Problem description and motivation:} Social sciences \cite{RW-JMM-RH-FRH-OM-RS-AG-PK-LH-RT:18}, psychology \cite{WJB-MJC-JJVB:20}, economy \cite{hiller} and control engineering \cite{AltafiniPlosOne} are all research areas that have a strong interest in understanding and describing opinion dynamics in social networks. As a consequence, many models of opinion dynamics have been proposed along time \cite{RH-UK:02}, \cite{NP-FD:19}. A problem of interest when dealing with social networks is the modelling and analysis of the spread of information in the network. Different works that address this problem and that focus on different diffusion mechanisms have been proposed, see \cite{RH-UK:02}, \cite{WM-PCV-GC-NEF-FB:17f}. 
A common  objective, in this context, is to understand when  reaching a \emph{consensus}, as a consequence of complex interactions among the agents in the network, is possible \cite{SRE-TB:15}.  Consensus is an active research topic in many fields \cite{CG-FB-JC-AJ:05q,GN-FB:06d}. It is  about the achievement of an agreement or of a common goal  by agents in a network. However, there are contexts in which the reaching of a consensus is either not desirable or does not represent a realistic scenario. This is the case also when dealing with social contexts, e.g. political elections, surveys. It is in these contexts that the \emph{disagreement} phenomenon, along with consensus, becomes of interest \cite{SRE-TB:15}.

A sociological model that considers both consensus and disagreement is the one known in the literature as Hegselmann-Krause (HK) model \cite{RH-UK:02}.
The HK dynamics evolves under a bounded-confidence mechanism. Confidence intervals are expressed as a function of the gap between pairs of agents' opinions. Since only agents whose opinions are close enough interact, the model represents a mathematical abstraction of \emph{confirmation bias} \cite{MDV-AS-GC-HES-WQ:17}. Confirmation bias is based on the natural human propensity to search for and welcome information that supports prior beliefs \cite{RSN:98}. In this paper we focus our attention on the Hegselmann-Krause model, assuming that agents are asked to express their opinion on a pre-fixed and finite number of topics. This represents an extension of the classical  scalar version \cite{RH-UK:02}. A  multi-dimensional version of the model has been already studied in the literature \cite{SRE-TB:15,SRE-TB-AN-BT:13,AN-BT:12}. It also finds a valid interpretation in other contexts such as robotics rendezvous problems 
\cite{FB-JC-SM:09}. The characterization of the dynamics in the multi-dimensional case is not trivial and some open questions still remain. In addition, we believe that the proposed multi-dimensional extension is not the only possible one. This is what motivates the work here presented.
\smallskip

\emph{Literature review:} The HK model is one of the first models that considers disagreement beside consensus. It represents an extension of the Friedkin and Johnsen model \cite{NEF-ECJ:99} to the case in which the topology of the network is opinion-variant. The model has attracted the interest of many researchers along time and many extensions have been proposed. 

The scalar model has been studied both in continuous \cite{JL:05a} and discrete time \cite{FV-CB-RI:21}. When all agents adopt the same confidence interval the model is said to be \emph{homogeneous}. Otherwise it is called \emph{heterogeneous}  \cite{GF-WZ-ZL:15}. If the lower and upper thresholds of acceptance of other agents' opinions are the same the model is called \emph{symmetric}. It is \emph{asymmetric} otherwise \cite{AB-MB-BC-HLN:13}. For the multi-dimensional case we mention the non-exhaustive list of works \cite{SRE-TB-AN-BT:13, AN-BT:12,SRE-TB:15, AB-MB-BC-HLN:13,SM-ET:14}. \\
In \cite{SRE-TB-AN-BT:13}  the authors focus on the investigation of the termination time of the dynamics. The analysis is based on Lyapunov arguments and  a polynomial upper bound for the case in which the connectivity of the network maintains some specific structure is provided. The work in \cite{AN-BT:12} focuses on the homogeneous multidimensional HK model, and assumes that  confidence intervals are expressed in terms of vector norms. Stability properties of the model are investigated and  the finite time convergence of the dynamics is proved. The results are valid regardless of the choice of the norm into play.
In \cite{SRE-TB:15} the evolution of the HK model under various assumptions is studied. First the termination time of the synchronous HK model in arbitrary finite dimensions is  analyzed and shown to be independent of the dimension of the opinion vectors.  The convergence speed of the dynamics is related to the eigenvalues of the adjacency matrix of the connectivity graph. A game-theoretic approach to the study of the asynchronous model is employed and  some results on the time of convergence of this variant are provided.  Finally, in the heterogeneous case   a necessary condition for the termination time to be finite is provided. 
 The work  \cite{AB-MB-BC-HLN:13} studies the case in which the confidence bound is determined through the Euclidean norm. The work focuses on the bounds on the convergence time of the system. The paper also investigates a noisy version of the model. 
In \cite{SM-ET:14} the   authors focus on heterophilious dynamics, namely on the tendency to be more welcoming towards different ideas rather than similar ones. The work shows a comparison among different functions that determine the  interactions between agents. They show that the heterophily mechanism enhances consensus more than the homophily one.
\smallskip

\emph{Contributions:} In this paper we consider two multi-dimensional  HK models for opinion dynamics:  the average-based model and the uniform affinity model.
  In the average-based mode that,  to the best of our knowledge, has not been considered before in the literature, individuals compare their average opinions on the various topics with those of the other individuals, and interact only with those individuals whose average opinions lie within a confidence interval. For this model we   provide an alternative proof for the contractivity of the range of opinions, and show that the agents'  opinions reach consensus/clustering if and only  if their average opinions  do so. 
 The uniform affinity model is a special instance of the multi-dimensional HK model investigated in \cite{AB-MB-BC-HLN:13,SRE-TB:15, SRE-TB-AN-BT:13, AN-BT:12}, where we specifically adopt the $\ell_\infty$-norm. In other words,  agents compare their opinions on each single topic and influence each other only if, topic-wise,  such opinions do not differ more than a given tolerance.  We   identify conditions under which the  uniform affinity model enjoys the order-preservation property topic-wise and we prove that  the global range of opinions (and hence the range of opinions on each single topic) are non-increasing.  
\smallskip

\emph{Paper organization:} Section II introduces some notation and prelimianry definitions and results. Section III to VI investigate the opinion ranges and the steady state behavior of the avergae-based HK model.
Finally, Section VII introduces the uniform affinity model and investigates some conditions that ensure the order  preservation of the opinions on each single topic.

\section{Notation and Preliminaries} 
In the following, $\realnonnegative$ denotes the set of nonnegative real numbers.
We let $\vect{1}_N$ and $\vect{0}_N$ denote the $N$-dimensional vectors of all ones and all zeros, respectively. The symbol $\vect{e}_i$ denotes the $i$-th vector of the canonical basis of $\real^N$, where $N$ will be clear from the context. 
Given a set ${\mathcal S}$, we denote by $|{\mathcal S}|$ its cardinality. \\
 Consider the complete undirected graph with $N$ nodes ${\mathcal G}=({\mathcal V}, {\mathcal E})$, where ${\mathcal V}=\{1,2,\dots, N\}$ and ${\mathcal E}= ({\mathcal V}\times {\mathcal V}) \setminus\{(i,i): i\in {\mathcal V}\}$.  If $|{\mathcal E}| =m$, we let $C_N\in \{-1, 0, 1\}^{N\times m}$ denote its {\em oriented incidence matrix}   \cite{FB:20}, defined as follows.
 For every vertex $h\in \mathcal{V}$ and every edge $e=(i,j) \in {\mathcal E}$, we have 
 $$[\HKn]_{h,e} = \begin{cases}
 1, & h=i;\\
 -1,& h=j;\\
 0,& {\rm otherwise}.\end{cases}
 $$
 Given a matrix $X\in \real^{N \times N}$, we denote by $X_{i*}$ the $i$-th row of $X$, by $X_{*j}$ its $j$-th column, and by $X_{ij}$ its $(i,j)$-th entry.

\begin{mydef}[Vector $\ell_\infty$-norm]
Given $x\in\real^N$,   the \emph{$\ell_\infty$-norm} of $x$ is
  $\norm{x}{\infty} =
  \max_{i} \lvert x_i \rvert.$
\end{mydef}
\begin{mydef}[Seminorms]
  A function $\map{\seminorm{ \cdot} }{\real^N}{\realnonnegative}$  is a \emph{seminorm} on $\real^N$ if it satisfies the following properties:
\begin{align*}
   & \text{(homogeneity): } \seminorm{ a x }= \lvert a
  \rvert\seminorm{ x },  \forall\ x\in \real^N \text{
    and } a \in \real;\\ & \text{(subadditivity): } \seminorm{ x + y}
  \leq \seminorm{ x}  + \seminorm{ y} ,  \forall \ x, y\in \real^N.
\end{align*}
\end{mydef}
\begin{mydef}\emph{($\ell_\infty$ weighted seminorm \cite{SJ-PCV-FB:19q+arxiv})}
  \label{semi_weighted}
  Let 
  $\map{\norm{ \cdot }{\infty}}{\real^k}{\realnonnegative}$
be the $\ell_\infty$-norm on $\real^k$ and  let $R \in
  \real^{k \times N}$. The \emph{$R$-weighted seminorm on $\real^N$
    associated with the $\ell_\infty$-norm on $\real^k$ is}
  \begin{equation}
    \seminorm{ x }_{\infty, R} := \norm{ Rx }{\infty},  \quad \forall x \in
    \real^N. \notag
  \end{equation}
\end{mydef}

\begin{myexp}[$\HKn^\top$-weighted  seminorm \cite{GDP-FB-MEV:arXiv}]\label{HK_wei}
Given a vector $x~\in~\real^N$  and  the oriented incidence matrix $\HKn \in \real^{N \times m}$, the $\HKn^\top$-weighted seminorm of $x$ associated with the $\ell_\infty$-norm is 
\begin{align*}
    \seminorm{ x}_{\infty,  \HKn^\top}
     = \max_{i, j} | x_i-x_j |. 
\end{align*}
\end{myexp}

\begin{mylem}[Preliminary lemma]\label{submult}
 Given a vector $x\in {\mathbb R}^N$ and a row stochastic matrix $A\in {\mathbb R}^{N \times N}$, 
$$\seminorm{Ax}_{\infty, C_N^\top} \le \seminorm{A}_{\infty, C_N^\top} \ \seminorm{x}_{\infty, C_N^\top},$$
where
$$\seminorm{A}_{\infty, C_N^\top} := \max_{\substack{\seminorm{x}{\infty, C_N^\top} =1\\ x \perp {\rm ker}(C_N^\top)}}\seminorm{Ax}_{\infty,C_N^\top}$$
is the $\HKn^\top$-weighted, $\ell_\infty$ induced seminorm of $A$.
\end{mylem}
\begin{proof} Upon noticing that 
$A  ({\rm ker} C_N^\top) \subseteq  {\rm ker} C_N^\top,$
the result follows from Lemma 14 in \cite{GDP-FB-MEV:arXiv} and the conditional sub-multiplicativity property of the semi norms \cite{VVK:83}.
\end{proof}

\begin{mythe}\label{mythe:seminorm}(\emph{Expression for the  $\HKn^\top$-weighted, $\ell_\infty$ induced seminorm \cite{GDP-FB-MEV:arXiv}})
For a row stochastic matrix $A \in \real^{N \times N}$,  
\begin{equation}\label{eq:infty_weighted}
\begin{aligned}
     \seminorm{ A }{\infty,  \HKn^\top}  = 1-\min_{ij}\sum_{k=1}^N {\rm min} \{A_{ik},  A_{jk}\}.
\end{aligned}
\end{equation}
\end{mythe}

\section{The  average-based (multi-dimensional)  HK model}

In this section we introduce a multi-dimensional extension of the HK model in which agents compare their (scalar) average opinions on a set of topics, rather than (the vectors representing) their specific opinions topic by topic. This model is suitable to describe the situation when   the opinions that an agent has on the different topics are not too far apart, as it happens, for instance, when the topics are related and homogeneous.
\\
Given a group of $N \geq 2$ agents and $m \ge 2$ (related) topics, we let $X_{ij}(t)$ denote the {\em opinion} that agent $i$ has about the topic $j$ at the time instant $t$.
The \emph{average opinion} that the agent $i$ has about the $m$ topics at the time instant $t$ is given by
\begin{equation*}
    \bar{x}_i(t) = \frac{1}{m} \sum_{j=1}^m X_{ij}(t)
\end{equation*}
and in vector form 
\begin{equation}\label{eq:xbar}
    \bar{x}(t) = \frac{1}{m} X(t) \vect{1}_m.
\end{equation}
We assume that the opinion that the $i$-th agent has on topic $j$ at time $t+1$ is influenced only by the opinions at time $t$ on that same topic of agents
whose  average opinion about the $m$ topics is not too far from agent $i$'s average opinion at time $t$. Specifically, given a certain {\em confidence threshold} $\varepsilon >0$, we define the set of {\em neighbours} (or influencers) of the agent $i$ at the time instant $t$ as a function of the average opinions of the agents, namely as:
\begin{equation}
    \mathcal{N}_i^{\rm ave}(\bar{x}(t)) = \{k \in \{1,\dots,N\}:  |\bar{x}_k(t)-\bar{x}_i(t)| \le \varepsilon\}.
\end{equation}
Accordingly,  by adopting a notation similar  to the one in \cite{RP-MF-BT:18},
 the \emph{influence matrix} $\Phi^{\rm ave} \in \{0,1\}^{N \times N}$ of this {\em average-based HK model} is defined as
\begin{equation}
    \Phi^{\rm ave}_{ik}(\bar{x}(t)) :=
    \begin{cases}
    1, & \text{ if } k\in \mathcal{N}_i^{\rm ave}(\bar{x}(t)); \\
    0, & \text{ otherwise}.
    \end{cases}
\end{equation}
Upon defining the matrix
\begin{equation}
    D^{\rm ave}(\bar{x}(t)) := 
    \begin{bmatrix}
    |\mathcal{N}_1^{\rm ave}(\bar{x}(t))| & & \\
    & \ddots & \\
    & & |\mathcal{N}^{\rm ave}_N(\bar{x}(t))|
    \end{bmatrix},
\end{equation}
the opinion matrix $X(t)$ evolves over time as
\begin{equation}\label{eq:sys}
    X(t+1) = A^{\rm ave}(\bar x(t))X(t),
\end{equation}
where
$$A^{\rm ave}(\bar x(t)) := D^{\rm ave}(\bar{x}(t))^{-1}\Phi^{\rm ave}(\bar{x}(t))$$
 is well-posed since $D^{\rm ave}(\bar x(t))$ is nonsingular as a consequence of the fact that $i\in \mathcal{N}_i^{\rm ave}(\bar x(t))$ (and hence $|\mathcal{N}_i^{\rm ave}(\bar{x}(t))| \ge 1$) $\forall i\in \{1,\dots,N\}$, $\forall t\geq 0$.
Equation \eqref{eq:sys} component-wise reads as
\begin{equation}
X_{ij}(t+1) = \frac{1}{|\mathcal{N}_i^{\rm ave}(\bar{x}(t))|} \sum_{k = 1}^n \Phi^{\rm ave}_{ik}(\bar{x}(t)) X_{kj}(t).
\end{equation}

\section{Average-based HK model: Main Definitions}
In this section we introduce some fundamental definitions for the average-based HK model that will be used in the following.

\begin{mydef}[Consensus for average-based HK model] The average-based HK model 
\begin{align}
    X(t+1)&= A^{\rm ave}(\bar x(t)) X(t),\label{eq:model1}\\
     \bar{x}(t) &= \frac{1}{m} X(t) \vect{1}_m,\label{eq:model2}
\end{align}
with $X(0) \in \real^{N\times m}$, 
is said to reach \emph{consensus} if
\begin{equation}
    \lim_{t \rightarrow\infty} X(t) = \vect{1}_N c^\top, \quad \exists\, c\in \real^m.
\end{equation}
\end{mydef}

\begin{mydef}[Clustering for average-based HK model]
The average-based HK model \eqref{eq:model1}-\eqref{eq:model2} reaches {\em clustering} if there exists a partitioning of the agents $\mathcal{V}_1, \mathcal{V}_2, \dots, \mathcal{V}_d$ ($\mathcal{V}_i \cap \mathcal{V}_j= \emptyset$ for $i \neq j$, and $\cup_{i=1}^d \mathcal{V}_i= \{1,\dots,N\})$ such that $\forall i,k \in \mathcal{V}_\ell,$ $\ell \in \{1,\dots, d\}$,
\begin{equation}
    \lim_{t \rightarrow \infty}  X_{i*}(t) = \lim_{t \rightarrow \infty}   X_{k*}(t) 
\end{equation}
and $\forall i\in \mathcal{V}_\ell$, $\forall k \in \mathcal{V}_p$, 
$\ell \neq p$,
\begin{equation}\label{eq:question}
    \lim_{t \rightarrow \infty} X_{i*}(t) \neq \lim_{t \rightarrow \infty} X_{k*}(t). 
\end{equation}
\end{mydef}
\begin{mydef}[Range of opinions on a specific topic] \label{ROST} Given the average-based HK model \eqref{eq:model1}-\eqref{eq:model2}, the \emph{range of opinions on topic $j$}, at the time instant $t$, is defined as
\begin{equation}
\nu_j(X(t)) = \max_{i,k\in \{1,\dots,N\}} |X_{ij}(t)-X_{kj}(t)|.
\end{equation}
\end{mydef}
\begin{remark} \label{njseminorm} Note that $\nu_j(X(t)) =  \seminorm{X_{*j}(t)}{\infty,C_N^\top}$.
\end{remark}
\begin{mydef}[$\varepsilon$-Chain for the average-based HK model] Consider the average-based HK model \eqref{eq:model1}-\eqref{eq:model2} and 
 assume that the entries of the average opinion vector in \eqref{eq:xbar} are ordered as $\bar{x}_1(t)\leq \dots \leq\bar{x}_N(t)$. The average opinion vector is an \emph{$\varepsilon$-chain} at the time instant $t$ if for all $k \in \{1,\dots, N-1\}$ $\bar x_k(t)-\bar x_{k+1}(t)\leq \varepsilon$.

\end{mydef}

\smallskip

\section{Average-based HK model: Opinion Ranges}

In this section we explore the monotonicity properties of the range of opinions defined in the previous section. As we will see, the average-based HK model preserves several nice properties of the scalar HK model \cite{RH-UK:02}.

\begin{myprop}[Range of opinions on topic] Given the  average-based HK model  \eqref{eq:model1}-\eqref{eq:model2}, for every choice of $X(0)\in {\mathbb R}^{N\times m}$, 
the range of opinions on a specific topic $\{\nu_j(X(t)\}_{t\ge 0}, j\in \{1,\dots,m\},$   is a non-increasing sequence.
\end{myprop}
\begin{proof} 
The proof 
 follows from the  fact that 
each column of $X(t)$ in \eqref{eq:model1} 
updates according to the equation
\begin{equation}
    X_{*j}(t+1)   = A(\bar x(t)) X_{*j}(t).
    \label{jcol_update}
\end{equation}
where $A(\bar x(t))$  is
row stochastic. 
\end{proof}
\begin{remark} By the same reasoning adopted to prove the previous result we can claim  that $\forall i \in \{1,\dots,N\},$ $j \in \{1,\dots,m\}$ and $t\geq 0$, one has  $X_{ij}(t) \in [\min_k X_{kj}(0), \max_k X_{kj}(0)]$. 
Consequently, if consensus is reached and we assume
 $c = [c_1\ \dots \ c_m]^\top$, then 
\begin{equation}
    \min_k X_{ki}(0) \leq c_i \leq \max_k X_{ki}(0).
\end{equation} 
\end{remark}
 In the following proposition we provide an alternative proof for the rate of contractivity of the range of opinions in the average-based HK model (see Remark \ref{remark12}, below). 
 
\begin{myprop}[Range of opinions]\label{myprop:range} Given the  average-based HK model  \eqref{eq:model1}-\eqref{eq:model2}, for every choice of $X(0)\in {\mathbb R}^{N\times m}$, 
 the range of opinions on a specific topic $\{\nu_j(X(t)\}_{t\ge 0}, j\in \{1,\dots,m\},$ 
 satisfies 
\begin{equation}
   \nu_j(X(t+1))  \leq \gamma(\bar x(t))  \nu_j(X(t)),
\end{equation}
where 
$$\gamma(\bar x(t)):= 1-\min_{i \ell}\sum_{k=1}^N \min\{A_{ik}(\bar{x}(t)),A_{\ell k}(\bar{x}(t))\}.$$
\end{myprop}
\begin{proof}
Consider \eqref{jcol_update}, where $A(\bar x(t))$ is row stochastic. From Remark \ref{njseminorm} and the submultiplicativity property of the induced matrix seminorms, we get 
\begin{align*}
&\nu_j(X(t+1))= \seminorm{X_{*j}(t+1)}{\infty,C_N^\top} \\
&=  \seminorm{A(\bar{x}(t))X_{*j}(t)}{\infty,C_N^\top} \leq\seminorm{A(\bar{x}(t))}{\infty,C_N^\top} \seminorm{X_{*j}(t)}{\infty,C_N^\top}\\
&=\Big( 1-\min_{i\ell}\sum_{k=1}^N \min\{A_{ik}(t),A_{\ell k}(t)\}\Big)\nu_j(X(t))
\end{align*}
where the  inequality follows from Lemma \ref{submult}, while the last identity from Theorem \ref{mythe:seminorm}. 
\end{proof}

\begin{myrem}\label{remark12}
The result in Proposition \ref{myprop:range} has also been proven in Lemma 1 in \cite{UK:00} Theorem 2.3 in \cite{SM-ET:14}  and Theorem 5.2 and Lemma 5.1 \cite{JS:00} for the scalar case. {\color{black}   Lemma 2.2 in \cite{RH-UK:19} pertains the multidimensional case.}   The proof in \cite{RH-UK:19} and  the one in \cite{UK:00} follow from arithmetic manipulation based on the row stochasticity of the matrix $A(t)$. The proof in  \cite{SM-ET:14} is related to a continuous time dynamics and the gap is computed with respect to a generic norm. The proof follows from the definition of dual norm. The proof in \cite{JS:00} exploits geometric considerations.
\end{myrem}
\smallskip

Opinions of the agents on each topic do not enjoy any order preservation property. 
  So, even if $\bar x_i(t) \le \bar x_j(t)$, nothing can be said about $X_{ik}(t)$ and $X_{jk}(t)$ for specific values of $k\in \{1,\dots,m\}$.
 For this reason,   agents whose opinions on a specific topic are very close may not influence each other.   Also, differently from the scalar case, there is no guarantee for order preservation among opinions.

\section{Average-based HK model: Steady State behavior}
We first note that the vector of the average opinions $\bar{x}(t)$ in \eqref{eq:xbar} obeys the dynamics
\begin{align}\label{eq:avemodel}
    \bar{x}(t+1) &= \frac{1}{m} X(t+1) \vect{1}_m  =\frac{1}{m} A^{\rm ave}(\bar{x}(t)) X(t) \vect{1}_m \notag \\
    &=A^{\rm ave}(\bar{x}(t)) \bar{x}(t), 
\end{align}
and hence it follows a scalar HK model.
 Upon a reordering of the agents, so that $\bar{x}_1(0) \leq \bar{x}_2(0) \leq \dots \leq \bar{x}_N(0)$ then, see Proposition 1 in \cite{VDB-JMH-JNT:09}, we can guarantee that $\bar{x}_1(t) \leq \bar{x}_2(t) \leq \dots \leq \bar{x}_N(t)$, $\forall t \geq 0$. Also, by Proposition 2 in \cite{VDB-JMH-JNT:09},    each sequence $\{\bar{x}_i(t)\}_{t\ge 0}$ is monotone and non-increasing and limited by $\bar{x}_N(0)$, therefore 
 $\lim_{t \rightarrow \infty} \bar{x}_i(t)$ exists and is finite, for every $i\in\{1,\dots,m\}$. Moreover, if $\bar{x}^* := \lim_{t \rightarrow \infty } \bar{x}(t)$,  $\forall i \in \{1, \dots, N-1\}$ either $\bar{x}_i^* = \bar{x}_{i+1}^*$ or $|\bar{x}_{i+1}^* - \bar{x}_{i}^*| > \varepsilon$, namely the steady state average opinions reach either consensus or clustering. Finally, according to Theorem 1 in \cite{VDB-JMH-JNT:09}, the limit configuration is reached in a finite number of steps, i.e., $\exists \, t^* \geq 0$ such that $\bar{x}(t^*) = \bar{x}^* \in \real^N$.

\begin{myrem}
As proved in \cite{JCD:01} consensus is reached if and only if the sequence $\bar{x}(t)$ is an $\varepsilon$-chain for all $t \geq 0$,   by this meaning that, assuming the initial ordering $\bar{x}_1(0) \leq \bar{x}_2(0) \leq \dots \leq \bar{x}_N(0)$, then we have $|\bar x_{i+1}(t)-\bar x_i(t)|\le \varepsilon$, for every $i\in \{1,\dots,N-1\}$ and $t\ge 0$.
\end{myrem}
\smallskip

Let us suppose that from $t^* \geq 0$ on-wards, 
\begin{equation}
    \bar{x}(t) =\frac{1}{m} X(t) \vect{1}_m = \bar{x}^* \in \real^N
\end{equation}
and hence
\begin{equation}
    X(t+1) = A^{\rm ave}(\bar{x}^*)X(t).
\end{equation}
Let us consider first, the case when $\bar{x}^* = c^* \vect{1}_N$, that is, $D^{\rm ave}(\bar{x}^*) = N I_N$ and $\Phi^{\rm ave}(\bar{x}^*) = \vect{1}_N \vect{1}_N^\top$.
Therefore, $\forall t \geq t^*$ we have $A^{\rm ave}(\bar{x}(t))= A^{\rm ave}(\bar{x}^*) = \frac{1}{N} \vect{1}_N \vect{1}_N^\top$  which is a constant doubly-stochastic symmetric matrix. Consequently, $\forall t \geq t^*$
\begin{equation}
    X(t+1) = \frac{1}{N} \vect{1}_N \vect{1}_N^\top X(t)
\end{equation}
which implies
\begin{equation}
    X(t^*+1) = \vect{1}_N [\bar{m}_1(t^*), \dots, \bar{m}_m(t^*)]
\end{equation}
where 
\begin{equation}
    \bar{m}_j(t^*) := \frac{1}{N} \sum_{i=1}^N X_{ij}(t^*)
\end{equation}
represents the average opinion of the agents on the $j$-th topic at the time instant $t^*$.
Consequently,
\begin{align}
    X(t^*+2) &= \frac{1}{N} \vect{1}_N\vect{1}_N^\top X(t^*+1)\\
    &= \frac{1}{N} \vect{1}_N\vect{1}_N^\top (\vect{1}_N [\bar{m}_1(t^*), \dots, \bar{m}_m(t^*)])\\
    &= \vect{1}_N[\bar{m}_1(t^*), \dots, \bar{m}_m(t^*)],
\end{align}
which means that if the HK model that describes the evolution of the average opinions of the agents on the $m$ topics reaches consensus at the time instant $t^*$, then   the punctual opinions of the agents on the topics reach consensus   at the next time-step
$$
     X(t) = X(t^*+1) = \vect{1}_N [m_1(t^*), \dots, m_m(t^*)], \ \forall t\ge t^*+1.
$$

Let us consider now the case when there exists $t^*$ such that $\bar{x}(t^*) = [c_1^* \vect{1}^\top_{n_1} | c_2^* \vect{1}^\top_{n_2} | \dots | c^*_d \vect{1}^\top_{n_d}]^\top$, namely the mean values of the agents' opinions on the $m$ topics clusterize into $d$ disjoint clusters:
$\mathcal{V}_1, \dots,\mathcal{V}_d$, $|\mathcal{V}_i| = n_i$, in each of which the average opinion takes value $c_i^*$ and $|c_i^*-c_{i+1}^*| > \varepsilon, \forall \, i \in \{1,\dots,N-1\} $. In this case, the matrices $D^{\rm ave}(\bar{x}^*)$ and $\Phi^{\rm ave}(\bar{x}^*)$  in \eqref{eq:model1}-\eqref{eq:model2} take the structure
\begin{align}
    &D^{\rm ave}(\bar{x}^*) = 
    \begin{bmatrix}
    n_1 I_{n_1} & & \\
    & \ddots & \\
    & & n_d I_{n_d}
    \end{bmatrix} \notag\\
        &\Phi^{\rm ave}(\bar{x}^*) = 
    \begin{bmatrix}
    \vect{1}_{n_1} \vect{1}_{n_1}^\top & & \\
    & \ddots & \\
    & & \vect{1}_{n_d} \vect{1}_{n_d}^\top
    \end{bmatrix}
\end{align}
and for all $t \geq t^*$
\begin{equation}
X(t+1)=
        \begin{bmatrix}
    \frac{1}{n_1}\vect{1}_{n_1} \vect{1}_{n_1}^\top & & \\
    & \ddots & \\
    & & \frac{1}{n_d}\vect{1}_{n_d} \vect{1}_{n_d}^\top
    \end{bmatrix} X(t).
\end{equation}
Consequently, 
\begin{equation}
X(t^*+1)=        \begin{bmatrix}
    \vect{1}_{n_1} & & \\
    & \ddots & \\
    & &  \vect{1}_{n_d}
    \end{bmatrix} M(t^*)
\end{equation}
with  
$M(t^*) \in {\mathbb R}^{d \times m}$ and 
\begin{align*}
   \vect{e}_i^\top M(t^*) &= \frac{1}{n_i} [\vect{0}^\top | \vect{1}_{n_i}^\top | \vect{0}^\top] X(t^*)= \frac{1}{n_i} \sum_{k \in I_i} \vect{e}_k^\top X(t^*), 
\end{align*}
 where $I_i = \{n_1+ \dots + n_{i-1}+1,\dots, n_1+ \dots + n_{i-1}+n_i\}$ is the set of agents in the $i$-th cluster. The $j$-th entry of the row vector $\vect{e}_i^\top M(t^*) \in \real^{1\times m}$ represents the  average opinion on the $j$-th topic of the agents in the $i$-th cluster.
Moreover,
\begin{align*}
X(t^*+2)&=    \begin{bmatrix}
    \frac{1}{n_1}\vect{1}_{n_1} \vect{1}_{n_1}^\top & & \\
    & \ddots & \\
    & & \frac{1}{n_d}\vect{1}_{n_d} \vect{1}_{n_d}^\top
    \end{bmatrix}  \cdot\\
    &\cdot \begin{bmatrix}
    \vect{1}_{n_1} & & \\
    & \ddots & \\
    & &  \vect{1}_{n_d}
    \end{bmatrix} M(t^*)\\
     &=\begin{bmatrix}
    \vect{1}_{n_1} & & \\
    & \ddots & \\
    & &  \vect{1}_{n_d}
    \end{bmatrix} M(t^*) = X(t^*+1).
\end{align*} 
Therefore, if the average opinions clusterize at the time instant $t^*$ then, from $t^*+1$ onward, the punctual opinions clusterize as well by maintaining the same partition, in $d$ clusters, as the average opinions of the agents over the $m$ topics. 
Note that
$$M(t^*) \vect{1}_{m}= \begin{bmatrix} c_1^*  & c_2^* &  \dots & c^*_d \end{bmatrix}^\top.$$

\begin{myrem}
If the average opinion vector $\bar{x}(t)$ clusterizes in $d$ clusters $\mathcal{V}_1,\mathcal{V}_2, \dots, \mathcal{V}_d$ then the opinions on each single topic $j$ clusterize in $d_j\leq d$ clusters and each cluster, say $\tilde{\mathcal{V}}_i$, is the union of one or more clusters  $\mathcal{V}_1,\mathcal{V}_2, \dots, \mathcal{V}_d$.
\end{myrem}

To summarize the results of this section we propose the following theorem.

\begin{theorem}[Steady state of average-based HK model]
Given the  average-based HK model  \eqref{eq:model1}-\eqref{eq:model2}, for every choice of $X(0)\in {\mathbb R}^{N\times m}$ the systems dynamics reaches a steady state configuration in a finite number of steps. Moreover, the average-based HK model  reaches consensus (clustering) if and only if the HK model describing the evolution of the average opinions reaches consensus (clustering).
\end{theorem}

The following result shows that if the maximum gap between the average opinions does not change when moving from time $t$ to time $t+1$, then the same maximum gap remains at all subsequent times, thus showing that if such gap is nonzero then consensus is not reached.

\begin{myprop}
Consider the  average-based HK model  \eqref{eq:model1}-\eqref{eq:model2}.
If at some time $t\geq 0$ one gets
\begin{equation}\label{myeq:maxequality}
    \max_{ij \in \{1,\dots,N\}}\!\! |\bar{x}_i(t)-\bar{x}_j(t)| =\!\! \max_{ij \in \{1,\dots,N\}}\!\!\! |\bar{x}_i(t+1)-\bar{x}_j(t+1)|
\end{equation}
then
\begin{equation}\label{myeq:maxequality2}
    \max_{ij \in \{1,\dots,N\}}\!\!\! |\bar{x}_i(t+1)-\bar{x}_j(t+1)| =\!\!\!\!\!\! \max_{ij \in \{1,\dots,N\}}\!\!\! |\bar{x}_i(t+2)-\bar{x}_j(t+2)|.
\end{equation}
Therefore,   if the quantity in \eqref{myeq:maxequality} is positive, then
the    average-based HK model  \eqref{eq:model1}-\eqref{eq:model2} does not achieve consensus.
\end{myprop}
\begin{proof}
Assume, without loss of generality, that $\bar{x}_1(t) \leq \bar{x}_2(t) \leq \dots \leq \bar{x}_N(t)$, then one has $\max_{ij\in \{1,\dots,N\}}|\bar{x}_i(t)-\bar{x}_j(t)| = \bar{x}_N(t)-\bar{x}_1(t)$.
Since  $\bar x_1(t+1) \geq \bar x_1(t)$ and  $\bar x_{N}(t+1) \leq \bar x_{N}(t)$, then \eqref{myeq:maxequality} implies 
$\bar x_1(t+1) = \bar x_1(t)$ and  $\bar x_{N}(t+1) = \bar x_{N}(t)$, that easily leads to \eqref{myeq:maxequality2}. Since the sequence of average opinions does not reach consensus, neither does the sequence $\{X(t)\}_{t\ge 0}$.
\end{proof}

\section{The uniform affinity model}

The multi-dimensional HK model investigated in \cite{SRE-TB:15,SRE-TB-AN-BT:13,AN-BT:12} has a  structure similar to the one we explored in the previous sections, however it adopts as a   criterion   to define the opinion proximity the distance (induced by the norm) between the opinion vectors of the agents.
Specifically, it is assumed that the neighbours   of agent $i$ at time $t$ are\footnote{Since the norm is formally defined for column vectors, while $X_{i*}(t)$ and $X_{k*}(t)$ are row vectors, we moved to their transposed versions.}
\begin{equation*}
    \mathcal{N}_i(X(t)) = \{k \in \{1,\dots,N\}: \norm{ X_{i*}(t)^\top  -  X_{k*}(t)^\top} \le \varepsilon\},
\end{equation*}
where $\varepsilon>0$ is the confidence threshold and $\norm{\cdot}$ denotes an arbitrary norm.
Accordingly, the  influence matrix $\Phi\in \{0,1\}^{N \times N}$ at time $t$ is the one whose $(i,k)$-th entry is
\begin{equation}
    \Phi_{ik}(X(t)) =
    \begin{cases}
    1, & \text{if } k \in \mathcal{N}_i(X(t));\\
    0, & \text{otherwise}.
    \end{cases}
\end{equation}
Upon defining the matrix\footnote{In the following we will replace ${\mathcal N}_i(X(t))$ with the more compact notation ${\mathcal N}_i(t).$  If we assume that $X(0)$ is assigned, the notation makes perfect sense.}
\begin{equation}
    D(X(t)) := 
    \begin{bmatrix}
    |\mathcal{N}_1(X(t))| & & \\
    & \ddots & \\
    & & |\mathcal{N}_N(X(t))|
    \end{bmatrix}
\end{equation}
the opinion matrix $X(t)$ evolves over time as
\begin{equation}\label{eq:sys_ta}
    X(t+1) = A(X(t))X(t),
\end{equation}
where
\begin{equation}
    \label{matriceA}
A(X(t)) := D(X(t))^{-1}\Phi(X(t))
\end{equation} 
  is well-posed {\color{black}($D(X(t))$ is nonsingular)  and row stochastic.}
  \\
  In the references \cite{SRE-TB:15,SRE-TB-AN-BT:13,AN-BT:12} the main focus has been on proving that the multi-dimensional HK model \eqref{eq:sys_ta}, with the row stochastic matrix $A(X(t))$ defined as above, (for any choice of the norm $\norm{\cdot}$) converges to  a steady-state solution in a finite number of steps, and on providing an upper bound on the termination time (see, in particular,
  \cite{SRE-TB:15}). The interesting aspect is that the termination time is independent of the number $m$ of topics. See Figure \ref{Fig1} for an example of an uniform affinity model with $N=10$ agents and $m=2$ topics that reaches consensus.\\
In this section we want to explore some monotonicity properties of the previous model by considering specifically the case when the norm is the $\ell_\infty$-norm. This means that
\begin{equation*}
    \mathcal{N}_i(X(t)) = \{j 
    : \max_{k\in \{1,\dots,m\}} | X_{ik}(t)  -  X_{jk}(t) | \le \varepsilon\}
\end{equation*}
 so, in order for two agents to influence each other, their opinions must be close topic-wise.  This model is in line with the spirit of bounded-confidence even in contexts in which agents  take different positions about the various topics. 
We will refer to the multi-dimensional HK model with $\ell_\infty$-norm \eqref{eq:sys_ta} as    the uniform affinity model.  \\
We first prove that if we consider the range of opinions on a specific topic $k$ at time $t$ and we consider the largest of such values over all the possible topics, then such a quantity is non increasing over time. 

\begin{myprop}{\label{myprop:gap}}(\emph{Range of opinions in uniform affinity HK model})
For the uniform affinity model, the quantity     
\begin{align*}
\nu(X(t)) &:= \max_{\substack{i,j \in \{1,...,N\}\\ k \in \{1,...,m\}}}  |X_{ik}(t)-X_{jk}(t)| \\
&= \max_{k \in \{1,...,m\}} \nu_k(X(t))
\end{align*}
is non increasing over time, namely $\nu(X(0)) \geq \nu(X(1)) \geq \nu(X(2)) \geq \dots$.
\end{myprop}
\begin{proof}
We first observe that $\forall {k} \in \{1,\dots,m\}$ 
\begin{align}
\nu(X(t)) &\geq \max_{ij}|X_{ik}(t)-X_{jk}(t)|\\ &= \max_i X_{i k}(t) - \min_j X_{jk}(t)\\
&= X_{uk}(t)-X_{lk}(t)
\end{align}
for some specific $u, l$.\\
For all $i,j \in \{1,\dots,N\}$ and   $k \in \{1,\dots,m\}$
\begin{align*}
    &|X_{ik}(t+1)-X_{jk}(t+1)| = \\
    &\Big| \sum_{d \in \mathcal{N}_i(t)} \frac{1}{|\mathcal{N}_i(t)|}X_{dk}(t) - \sum_{d \in \mathcal{N}_j(t)} \frac{1}{|\mathcal{N}_j(t)|}X_{dk}(t) \Big|\leq\\
    & \Big| \max_\ell X_{\ell k}(t) - \min_\ell X_{\ell k}(t) \Big| = \\
    & \Big| X_{u k}(t) -  X_{l k}(t) \Big| = X_{uk}(t)-X_{lk}(t) \leq \nu(X(t)).
\end{align*}
Since this is true for all $i,j \in \{1,\dots,N\}$ and for all $k \in \{1,\dots,m\}$, then it is also true that
\begin{equation*}
    \nu(X(t+1)) = \max_{\substack{ij\in \{1,\dots,N\}\\ k \in \{1,\dots,m\}}} |X_{ik}(t+1)-X_{jk}(t+1)|  \leq \nu(X(t)).
\end{equation*}
\end{proof}
\smallskip

\begin{figure*}[t]
    \centering
    \begin{minipage}{\columnwidth}
    \includegraphics[scale=0.29]{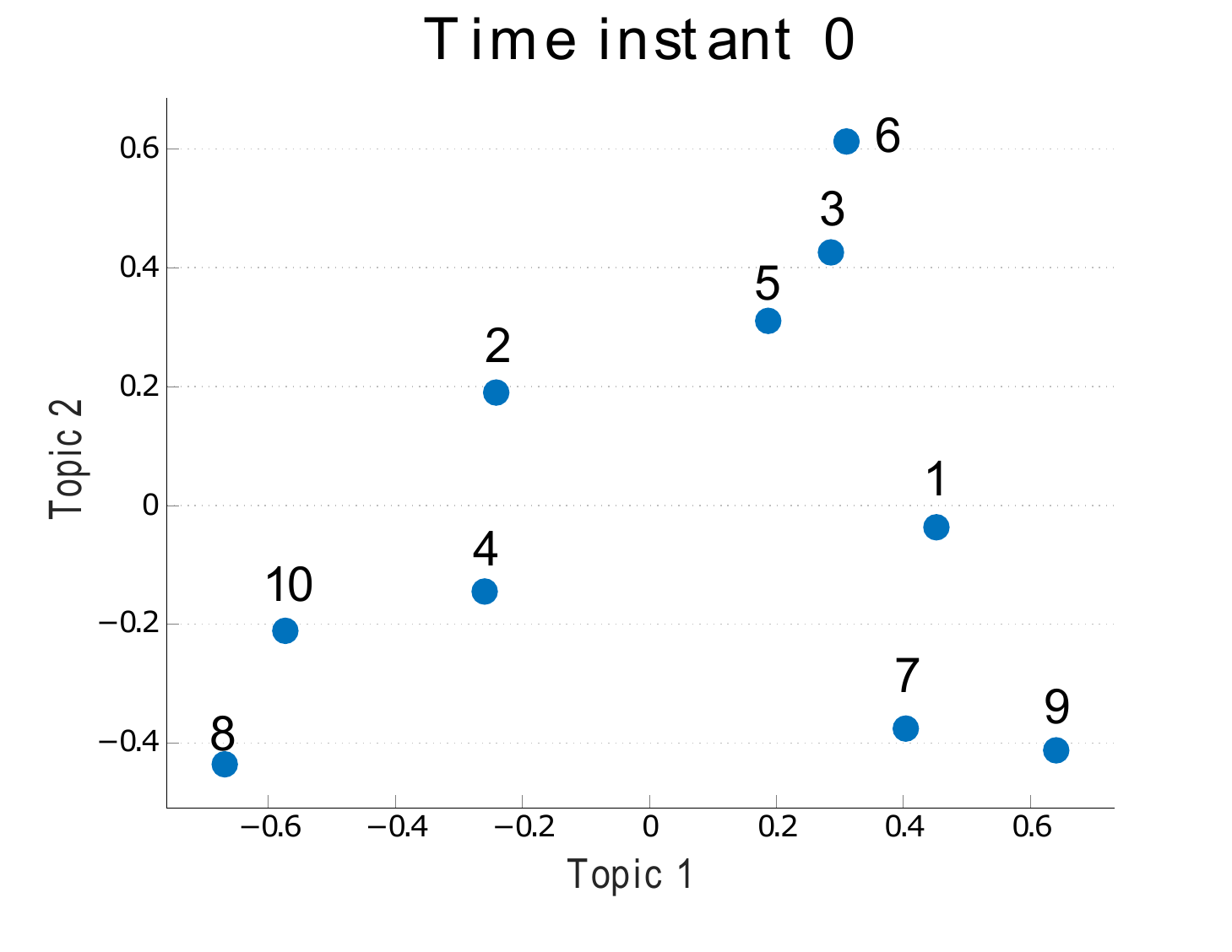}\hfill
     \includegraphics[scale=0.29]{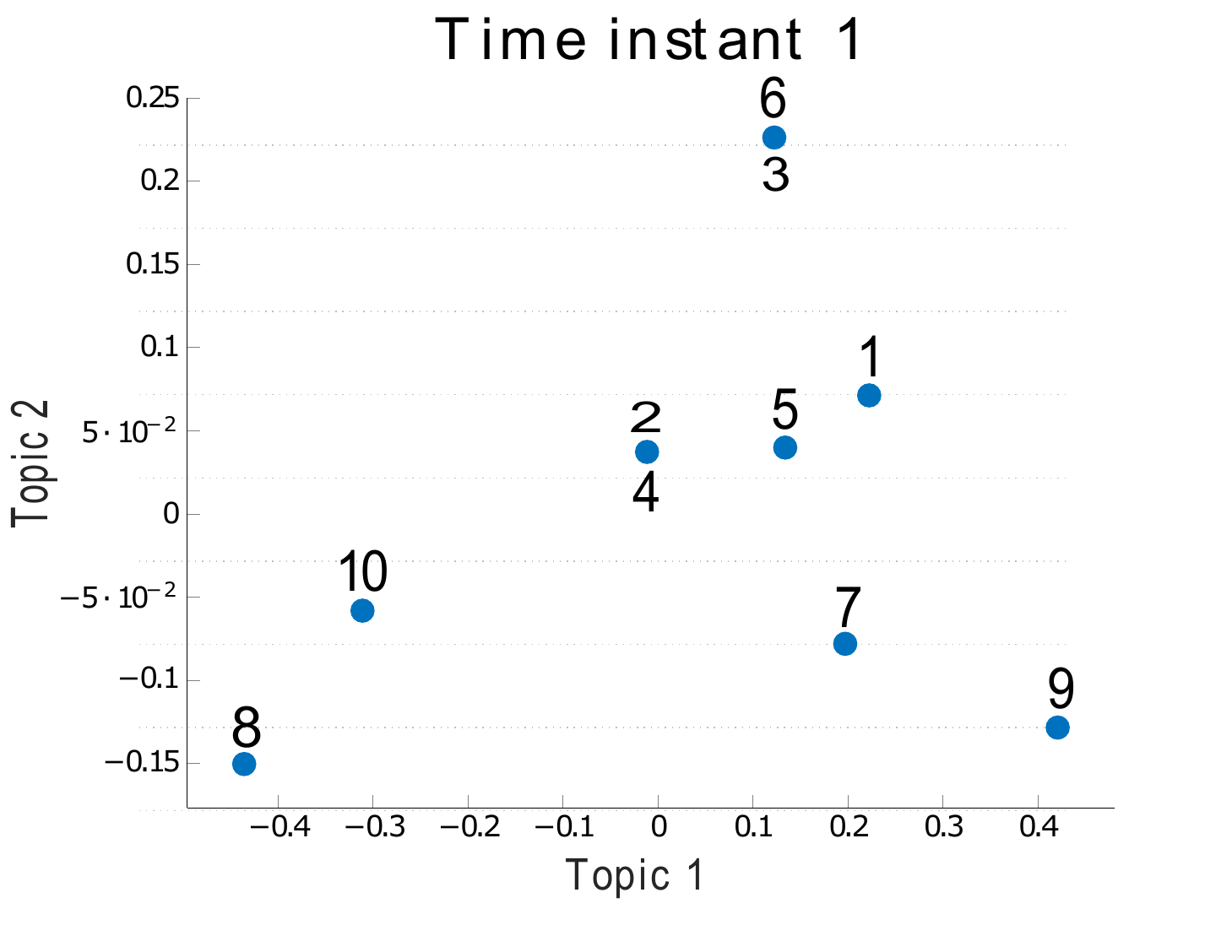}
   
    \end{minipage}\hfill
    \begin{minipage}{\columnwidth}
     \includegraphics[scale=0.29]{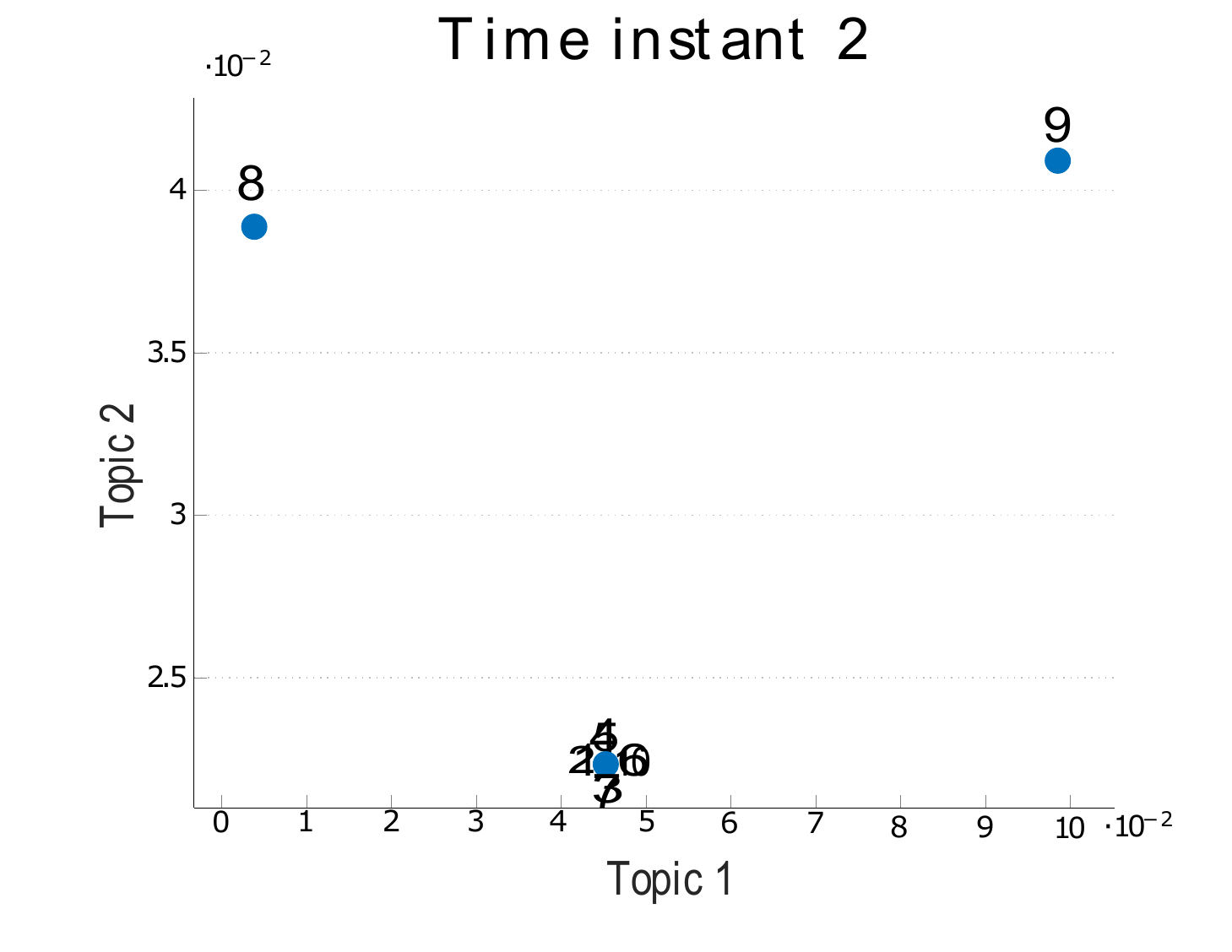}\hfill
    \includegraphics[scale=0.29]{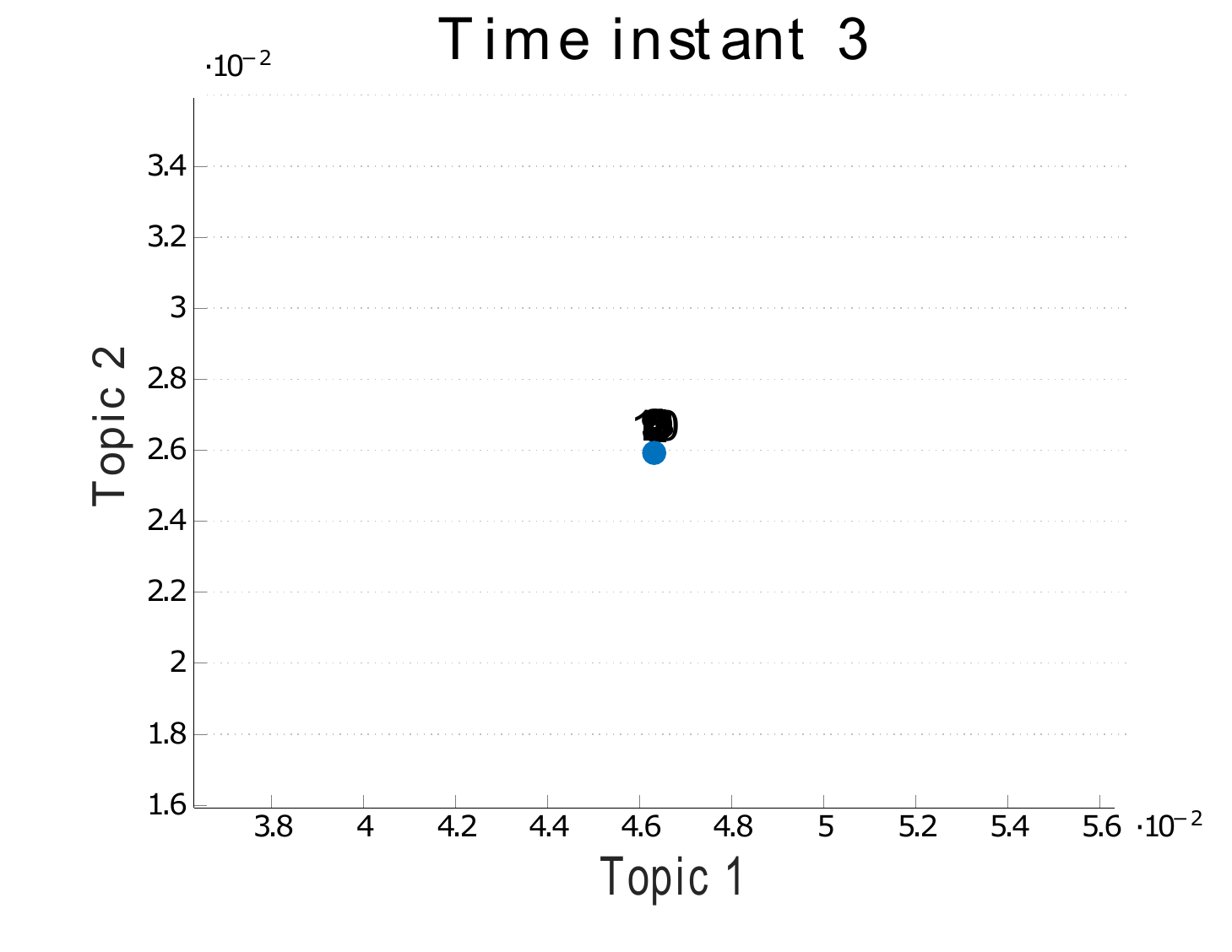}
    \end{minipage}
    \caption{Convergence to consensus of the uniform affinity model with $N=10$ agents, $m=2$ topics, confidence threshold $\varepsilon = 0.8$. Initial conditions are uniformly generated in the interval $[-1,1]^m$. Convergence occurs after $3$ iterations.}
    \label{Fig1}
\end{figure*}

 A consequence of Proposition \ref{myprop:gap} is that the opinion gap on each single topic   (see Definition \ref{ROST}) is non-increasing too. 
\smallskip 

Differently from what happens with the standard HK model (and partly with the average-based HK model), there is no way to introduce a meaningful total ordering in $\mathbb{R}^m$ and hence in the set of all agents' opinions.
In fact, in general, the ordering is different on each topic, and condition $|X_{ik}(t) - X_{jk}(t)|\le \varepsilon$ for some specific $k$ does not ensure that $i$ and $j$ are neighbours. So, it may happen that $X_{ik}(t) < X_{jk}(t)$, but at the subsequent time step $X_{ik}(t+1) > X_{jk}(t+1).$  
See, for instance, Figure \ref{Fig1}, where moving from $t=0$ to $t=1$ the opinions of agents $1$ and $5$ on topic 2 swap  (even if eventually they all converge to a consensus). \\   However, if   at some time $t$  every pair of  agents who do not influence each other  have opinions about {\em all} the topics that differ by more than $\varepsilon$, then the opinions' ordering on each topic remains unaltered when moving from time $t$ to $t+1$.

 \begin{myprop}[One-step order preservation]\label{myprop:ordering}
Consider the uniform affinity model and suppose that at some time $t\ge 0$ 
one has that  for every $i,j\in\{1,\dots, N\}$ 
condition $j\not\in {\mathcal N}_i(t)$ implies \begin{equation}
|X_{ik}(t) - X_{jk}(t)| > \varepsilon,
\quad \forall\ k\in \{1,\dots,m\}.
    \label{tuttidistanti}
\end{equation}
    If for every $k\in \{1,\dots,m\}$ we   sort the agents' opinions (the order specifically depending on $k$) so that 
    $X_{i_1 k}(t)\le X_{i_2 k}(t) \le \dots \le X_{i_N k}(t),$ then the same opinion ordering on that topic is preserved at $t+1$, i.e.,
    $X_{i_1 k}(t+1)\le X_{i_2 k}(t+1) \le \dots \le X_{i_N k}(t+1).$
\end{myprop}
\begin{proof} 
Let $k$ be arbitrary in $\{1,\dots,m\}$ and consider $i_h \in\{i_1, \dots, i_{N-1}\}$.\\
  We preliminarily observe that if $i_{h+1} \not\in {\mathcal N}_{i_h}(t)$, assumption \eqref{tuttidistanti}  implies that
\begin{align*}
&{\mathcal N}_{i_h}(t)\cap {\mathcal N}_{i_{h+1}}(t)=\emptyset;\\
&{\mathcal N}_{i_h}(t)\setminus {\mathcal N}_{i_{h+1}}(t) \subseteq \{i_1,\dots, i_h\};\\
&{\mathcal N}_{i_{h+1}}(t)\setminus {\mathcal N}_{i_{h}}(t) \subseteq \{i_{h+1},\dots, i_N\}
\end{align*}
On the other hand, if $i_{h+1} \in {\mathcal N}_{i_h}(t)$, then
\begin{align*}
&{\mathcal N}_{i_h}(t)\setminus {\mathcal N}_{i_{h+1}}(t) \subseteq \{i_1,\dots, i_{h-1}\};\\
&{\mathcal N}_{i_{h+1}}(t)\setminus {\mathcal N}_{i_{h}}(t) \subseteq \{i_{h+2},\dots, i_N\}.\\
\end{align*}
So,
 we can define
\begin{align*}
    \Delta_{hk}(t)&:= \frac{1}{|\mathcal{N}_{i_h}(t)\cap \mathcal{N}_{i_{h+1}}(t)|}\sum_{\ell \in \mathcal{N}_{i_h}(t)\cap \mathcal{N}_{i_{h+1}}(t)} X_{\ell k}(t),\\
    \tilde{X}_{i_hk}(t)&:= \frac{ \sum_{\ell \in \mathcal{N}_{i_h}(t)\setminus \mathcal{N}_{i_{h+1}}(t)} X_{\ell k}(t)}{|\mathcal{N}_{i_h}(t)\setminus \mathcal{N}_{i_{h+1}}(t)|},\\
    \tilde{X}_{i_{h+1}k}(t)&:= \frac{ \sum_{\ell \in \mathcal{N}_{i_{h+1}}(t)\setminus \mathcal{N}_{i}(t)} X_{\ell k}(t)}{|\mathcal{N}_{i_{h+1}}(t)\setminus \mathcal{N}_{i}(t)|},
\end{align*}
and get
\begin{align*}
    &X_{i_hk}(t+1)=\\
    &\frac{|\mathcal{N}_{i_h}(t)\cap \mathcal{N}_{i_{h+1}}(t)|}{|\mathcal{N}_{i_h}(t)|}\Delta_{hk}(t) + \frac{|\mathcal{N}_{i_h}(t)\setminus \mathcal{N}_{i_{h+1}}(t)|}{|\mathcal{N}_{i_h}(t)|}  \tilde{X}_{i_hk}(t)  
\end{align*}
and similarly
\begin{align*}
&X_{i_{h+1}k}(t+1)=\\
&\frac{|\mathcal{N}_{i_h}(t)\cap \mathcal{N}_{i_{h+1}}(t)|}{|\mathcal{N}_{i_{h+1}}(t)|}\Delta_{hk}(t)+ \frac{|\mathcal{N}_{i_{h+1}}(t)\setminus \mathcal{N}_{i_h}(t)|}{|\mathcal{N}_{i_{h+1}}(t)|}  \tilde{X}_{jk}(t).  \end{align*}
Since $\tilde{X}_{i_{h+1}k}(t) \ge \Delta_{hk}(t) \ge \tilde{X}_{i_hk}(t)$ if $i_h$ and $i_{h+1}$ are neighbours, while $\Delta_{hk}(t)=0$ and $\tilde{X}_{i_{h+1}k}(t) > \tilde{X}_{i_hk}(t)$ 
if $i_h$ and $i_{h+1}$ are not neighbours, it follows  that $\tilde{X}_{i_{h+1}k}(t+1) \ge  \tilde{X}_{i_hk}(t+1)$.

\end{proof}

The reasoning behind the previous result can be extended to a different situation when
 the agents' opinions at some time $t$ are ordered so that for every $k\in \{1,\dots,m\}$
$$X_{1k}(t) \le X_{2k}(t)\le\dots\le X_{Nk}(t).$$
When so, such ordering is preserved at all subsequent time instants. This is based on the fact that if $i < j$ and $i$ and $j$  are not neighbours,  then there exists $\bar k$ such that $X_{j\bar k}(t) - X_{i\bar k}(t) > \varepsilon$. But this implies that for every $p <i$ one has
$X_{j\bar k}(t) - X_{p\bar k}(t) > \varepsilon$, and for every $q > j$ one has
$X_{q\bar k}(t) - X_{i\bar k}(t) > \varepsilon$.
Consequently, ${\mathcal N}_i(t)\cap \{j, j+1, \dots, N\} =\emptyset$ and similarly
${\mathcal N}_j(t)\cap \{1, 2, \dots, i\} =\emptyset$. Conversely, if $i < j$ and $i$ and $j$  are  neighbours, then ${\mathcal N}_i(t)\cap {\mathcal N}_j(t) \supseteq \{i, i+1, \dots, j\}.$ Based on these comments, the proof of the following result can be easily obtained by mimicking the proof of 
Proposition \ref{myprop:ordering}.

\begin{myprop}[Sufficient condition for order preservation]\label{myprop:ordering0} Consider the uniform affinity model. If at some time $t\ge 0$ one has that for every topic $k\in \{1,\dots,m\}$
$$X_{1k}(t) \le X_{2k}(t)\le\dots\le X_{Nk}(t),$$
then it is also true that for every $\tau \ge 0$ and  every topic $k\in \{1,\dots,m\}$
$$X_{1k}(t+\tau ) \le X_{2k}(t+\tau)\le\dots\le X_{Nk}(t+\tau).$$
\end{myprop}

\section{Conclusions}

In this paper two multi-dimensional extensions of the Hegelsemann-Krause model have been proposed.
Multi-dimensionality refers to the fact that   agents are asked to express their opinion on multiple topics.

The average-based model assumes that two agents interact when their opinions lie, on average, within each others' confidence interval. This models fits to contexts in which the topics into play are somehow related, so that the agents are unlikely to  assume   far apart opinions about the various topics.
The uniform affinity model is a specific case of the multidimensional HK model proposed in \cite{SRE-TB-AN-BT:13,SRE-TB:15,AN-BT:12} in which the $\ell_\infty$ norm is exploited to measure the distance between opinions.
This model applies to more general contexts, in which the topics may be unrelated and each individual may take far apart positions about different topics.
It has been shown how, after some algebraic manipulation, the average-based model can be traced back to a standard HK model and consensus/clustering on each single topic is achieved if and only if consensus/clustering is achieved on average. For the uniform affinity model sufficient conditions for order preservation among opinions have been proposed.

\bibliographystyle{plainurl+isbn}
\bibliography{alias, Main, FB,BibHK, Refer166}

\clearpage


\end{document}